\documentclass[a4paper,UKenglish,cleveref,autoref,thm-restate]{lipics-v2021}
\usepackage[size=footnotesize]{todonotes}

\usepackage{uniinput}
\usepackage{siunitx}

\DeclareMathOperator{\tw}{tw}

\DeclareMathOperator{\cptw}{cptw}

\DeclareMathOperator{\val}{val}

\pdfoutput=1 
\hideLIPIcs  

\title{Partitioning the Bags of a Tree Decomposition Into Cliques}

\author{Thomas Bläsius}{Karlsruhe Institute of Technology, Germany \and \url{http://scale.iti.kit.edu} }{thomas.blaesius@kit.edu}{https://orcid.org/0000-0003-2450-744X}{}

\author{Maximilian Katzmann}{Karlsruhe Institute of Technology, Germany}{maximilian.katzmann@kit.edu}{https://orcid.org/0000-0002-9302-5527}{}

\author{Marcus Wilhelm}{Karlsruhe Institute of Technology, Germany}{marcus.wilhelm@kit.edu}{https://orcid.org/0000-0002-4507-0622}{}

\authorrunning{T. Bläsius, M. Katzmann, M. Wilhelm} 



\supplement{\url{https://github.com/marcwil/cptw_code}}

\nolinenumbers 

\EventEditors{John Q. Open and Joan R. Access}
\EventNoEds{2}
\EventLongTitle{42nd Conference on Very Important Topics (CVIT 2016)}
\EventShortTitle{CVIT 2016}
\EventAcronym{CVIT}
\EventYear{2016}
\EventDate{December 24--27, 2016}
\EventLocation{Little Whinging, United Kingdom}
\EventLogo{}
\SeriesVolume{42}
\ArticleNo{23}

\begin{document}

\maketitle

\begin{abstract}
  We consider a variant of treewidth that we call
  \emph{clique-partitioned treewidth} in which each bag is
  partitioned into cliques.  This is motivated by the
  recent development of FPT-algorithms based on similar
  parameters for various problems.  With this paper, we
  take a first step towards computing clique-partitioned
  tree decompositions.

  Our focus lies on the subproblem of computing clique
  partitions, i.e., for each bag of a given tree
  decomposition, we compute an optimal partition of the
  induced subgraph into cliques.  The goal here is to
  minimize the product of the clique sizes (plus 1).  We
  show that this problem is NP-hard.  We also describe four
  heuristic approaches as well as an exact branch-and-bound
  algorithm.  Our evaluation shows that the
  branch-and-bound solver is sufficiently efficient to
  serve as a good baseline.  Moreover, our heuristics yield
  solutions close to the optimum.  As a bonus, our
  algorithms allow us to compute first upper bounds for the
  clique-partitioned treewidth of real-world networks.  A
  comparison to traditional treewidth indicates that
  clique-partitioned treewidth is a promising parameter for
  graphs with high clustering.

\end{abstract}

\section{Introduction}

The treewidth is a measure for how treelike a graph is in
terms of its separators.  It is defined via a tree
decomposition, a collection of vertex separators
called \emph{bags} that are arranged in a tree structure.
The size of the largest bag determines the width of the
decomposition and the treewidth of a graph is the minimum
width over all tree decompositions.

The concept of treewidth has its origins in graph theory
with some deep structural
insights~\cite{Lovasz_2005,robertson_seymour_graphminors_iv}.
Additionally, there are algorithmic implications.
Intuitively speaking, the separators of a tree
decomposition split the graph into pieces that can be
solved independently except for minor dependencies at the
separators. This is often formalized using a dynamic
program over the tree decomposition, yielding an
FPT-algorithm (fixed-parameter tractable) with the
treewidth as parameter~\cite{Cygan2015}. As this is a
versatile framework that can be applied to many problems,
it comes to no surprise that there has been quite a bit of
effort to develop algorithms for computing low-width tree
decompositions (see,
e.g.,~\cite{pace_2016_report,pace_2017_report}).

A major obstruction for low treewidth are large cliques,
which inevitably lead to large separators. This is
particularly true for so-called complex networks, i.e.,
graphs with strong community structure and heterogeneous
degree distribution, which appear in various domains such
as communication networks, social networks, or
webgraphs. One could, however, hope for two aspects that
together mitigate this negative effect of large
cliques. First, though some separators need to be large,
these separators are structurally simple, e.g., they form a
clique or can be covered with few cliques. Second,
separators that are large but structurally simple still let
us solve the separated pieces individually with low
dependence between them. The first hope is supported by the
fact that the treewidth is asymptotically equal to the
clique number in hyperbolic random
graphs~\cite{blaesius_hyp_sep_treewidth}; a popular model
for complex networks~\cite{hyp_geom_compl_ntwrks}.  This
indicates that cliques are indeed the main obstruction for
low treewidth in these kinds of networks. The second hope
is supported by the results of de Berg et
al.~\cite{deBerg_framework_intersection}, who introduced
the concept $\mathcal{P}$-flattened tree
decompositions. There, the graph is partitioned into
cliques and the width of the tree decomposition is measured
in terms of the (weighted) number of cliques in a bag.
Thus, the width does measure the complexity of separators
rather than their size. Based on this definition, the
authors then show that these structurally simple separators
help to solve various graph problems efficiently.

To the best of our knowledge, these extended concepts have
not yet been studied from a practical perspective.  With
this paper, we want to initiate this line of research by
addressing two questions.  First, can such
clique-partitioned tree decompositions lead to
substantially smaller width values than classical tree
decompositions?  Second, how can such tree decompositions
be computed?  For the second question, we design and
evaluate different algorithmic strategies for computing a
novel yet closely related variant of tree decompositions.
Our experiments yield some interesting algorithmic insights
and provide a good starting point for further development.
On networks that do exhibit clique structures, the
constructed tree decompositions indeed have sufficiently
low width to answer the first question affirmatively.  We
believe that there is plenty of room for improvement in our
approaches, which may yield even better insights into the
applicability of the new parameter.  In the following, we
discuss related work before stating our contribution more
precisely.

\subsection{Related Work}

There are multiple lines of research that investigate
variants of treewidth where additional structural
properties are taken into account.  As mentioned above, De
Berg et al.~\cite{deBerg_framework_intersection} propose a
variant of tree decompositions where the initial graph is
partitioned into cliques (or unions of constantly many
connected cliques) that are contracted into weighted
vertices.  The weight of a clique of size $s$ is
$\log(s+1)$ and the weight of a bag of the tree
decomposition is the sum of its weights.  Using this
technique, they give subexponential algorithms for a range
of problems on geometric intersection graphs, including
\textsc{Independent Set}, \textsc{Steiner Tree} and
\textsc{Feedback Vertex Set}.  For some of these problems,
the algorithms are also representation agnostic, while for
most others, the geometric representation is required.
They also prove that the running time of the algorithms is
tight under the exponential-time-hypothesis (ETH).
Kisfaludi-Bak~\cite{bak_hyperbolic_intersection} applied
the same algorithmic framework to intersection graphs of
constantly sized objects in the hyperbolic plane.

A similar parameter called \emph{tree clique width} has
been proposed by Aronis~\cite{aronis_tree_clique_width}.
Here, the idea is to consider tree decompositions where
each bag is annotated with an edge clique cover (ecc) and
where the size of the cover determines the width of a bag.
The paper shows several hardness results and adapts common
treewidth algorithms to the newly proposed parameter.

Another approach to capture graph structures that lead to
high treewidth despite being structurally simple has been
proposed by Dallard, Milanič, and Štorgel.  They define the
\emph{independence number} of a tree decomposition as the
size of the largest independent set of any of its bags and
the \emph{tree-independence number} of a graph as the
minimum independence number of any tree decomposition
\cite{dallard_treewidth_clique_ii}.  This parameter
connects to the more theoretical study of
$(\mathrm{tw}, ω)$-bounded graphs, i.e., graph classes
in which the treewidth depends only on the clique number~\cite{dallard_treewidth_clique_i,
  dallard_treewidth_clique_iii}.  This line of research is
mostly concerned with the classification and
characterization of the considered graph classes both in
terms of graph theory and algorithmic exploitability.
However, apart from a factor $8$ approximation with running
time $2^{O(k^{2})}\cdot n^{O(k)}$ due to Dallard, Fomin,
Golovach, Korhonen, and
Milanič~\cite{dallard_computing_tree_ind_num}, we are not
aware of any work that tries to actually build algorithms
for this or similar parameters.

\subsection{Contribution}

In this paper, we propose \emph{clique-partitioned}
treewidth as a parameter that captures structurally simple
separators in graphs.  It can be seen as a close adaptation
of $\mathcal{P}$-flattened
treewidth~\cite{deBerg_framework_intersection}, where we
\emph{first} compute a tree decomposition and \emph{then}
determine clique partitions of the subgraphs induced by the
bags.  Thus, instead of using a global clique partition of
the whole graph, we consider clique partitions that are
local to a single bag.

The remainder of this paper is structured as follows.  In
\cref{sec:cptw}, we formalize our definition for
clique-partitioned treewidth and prove several statements
comparing it with $\mathcal{P}$-flattened treewidth.  In
\cref{sec:clique-partition}, we present multiple approaches
to compute low-weight clique partitions for the bags of a
tree decomposition.  They include various heuristic
methods, as well as an exact branch-and-bound algorithm for
which we propose several adjustments with the potential to
improve its running time in practice.  Afterwards, in
\cref{sec:eval} we combine an implementation of our
approaches with existing methods for computing tree
decompositions and study the upper bounds on the
clique-partitioned treewidth of real-world networks.
Furthermore, we evaluate the performance of the exact and
heuristic clique partition solvers proposed in
\cref{sec:clique-partition}.

\section{Clique-partitioned treewidth}
\label{sec:cptw}

We first introduce some basic notation and give the
definition for traditional tree decompositions.  We write
$[n] = \{1, \dots, n\}$ for the first $n$ natural numbers.
Throughout the paper, we assume graphs $G = (V, E)$ to be
simple and undirected and write $V(G)$ and $E(G)$ for the
sets of vertices and edges, respectively.  For a subset
$X\subseteq V$ we write $G[X]$ for the subgraph of $G$
induced by $X$.

A \emph{tree decomposition} of $G$ is a pair $(T, B)$, for
a tree $T$ and a function $B$ mapping vertices of $T$ to
subsets of $V$ called \emph{bags} such that $T$
and $B$ have the following three properties: (1) every
vertex of $G$ is contained in some bag, (2) for every edge,
there is a bag containing both endpoints, and (3) for any
vertex $v$ of $G$, the set of bags containing $v$ forms a
connected subtree of $T$.  The \emph{width} of a tree
decomposition is the size of the largest bag minus 1.  The
\emph{treewidth} $\tw(G)$ is the smallest width obtainable
by any tree decomposition of $G$.

We define a \emph{clique-partitioned} tree decomposition of
$G$ as a tree decomposition where for every $t \in V(T)$ we
have a partition $\mathcal{P}_{t}$ of the subgraph induced
by the corresponding bag (i.e., the graph $G[B(t)]$) into
cliques.  Following de Berg et
al.~\cite{deBerg_framework_intersection}, we define the
\emph{weight} of a clique $C$ as $\log(|C| +1)$ and the
weight of a bag $B(t)$ as the sum of weights of the cliques
in its partition $\mathcal{P}_{t}$.  Throughout this paper
we assume $2$ to be the default base of logarithms.  The
weight of a clique-partitioned tree decomposition is the
maximum weight of any of its bags and the
\emph{clique-partitioned treewidth} (short: cp-treewidth)
of $G$, denoted by $\cptw(G)$, is the minimum weight of any
clique-partitioned tree decomposition.

As mentioned before, the clique-partitioned treewidth is
closely related to the parameter defined by de Berg et
al.~\cite{deBerg_framework_intersection}.  For a clique
partition $\mathcal{P}$ of the whole graph $G$, we say that
a \emph{$\mathcal{P}$-flattened tree decomposition} is a
clique-partitioned tree decomposition of $G$ where the
partition into cliques within a bag is induced by the
global partition $\mathcal{P}$.  As before, the weight of a
$\mathcal{P}$-flattened tree decomposition is the maximum
total weight of the cliques in any of its bags.  In
reference to the
authors~\cite{deBerg_framework_intersection}, we call the
minimum weight over all $\mathcal{P}$ the BBKMZ-treewidth.

We note that our parameter can also be seen as an
adaptation of \emph{tree clique
  width}~\cite{aronis_tree_clique_width}, where instead of
considering the size of an \emph{edge clique cover} of each
bag, we consider the logarithmically weighted sum of clique
sizes of a clique partition.  That is, we are using the
weight function of the $\mathcal{P}$-flattened treewidth to
define a parameter which considers individual clique
partitions, similar to tree clique width.

In the following, we compare the clique-partitioned
treewidth to the more closely related
BBKMZ-treewidth. First, as a global partition $\mathcal{P}$
can also be used locally in each bag of a
clique-partitioned tree decomposition, we obtain that the
clique-partitioned treewidth of a graph is at most its
BBKMZ-treewidth.  Additionally, the clique-partitioned
treewidth can also be substantially smaller than the
BBKMZ-treewidth, as shown in the following lemma.

\begin{lemma}\label{lem:parameter-local-log-global}
  There is an infinite family of graphs $\mathcal{G}$ such
  that a graph $G \in \mathcal{G}$ with $n$ vertices, has
  clique-partitioned treewidth in
  $\mathcal{O}\left( \log\log n \right)$ and
  BBKMZ-treewidth in
  $\Omega\left(\log n\right)$.
\end{lemma}
\begin{proof}
  The family $\mathcal G$ contains for every
  $h \in \mathbb N$ one graph $G_h$.  The Graph $G_{h}$ is
  a complete binary tree of height $h$, where additionally
  for every leaf $\ell$ we connect all $h$ vertices that
  lie on a path between the root $r$ and $\ell$ into a
  clique.  Note that we have $h \in \Theta(\log n)$.

  Let $\mathcal{P}_{h}$ be a clique partition of $G_{h}$.
  Then, via a simple induction over $h$, it is easy to see
  that in $G_{h}$ there is a path between the root $r$ and
  some leaf $\ell$ of $G_{h}$ such that every vertex on the
  path belongs to a different partition class.  These
  vertices form a clique in $G_{h}$ that has to be prosent
  in some bag of any $\mathcal{P}_{h}$-flattened tree
  decomposition of $G_{h}$.  This bag thus contains all $h$
  partition classes on the path and has weight
  $h \cdot \log(1+1) \in \Omega(\log n)$.

  At the same time we can construct a clique-partitioned
  tree decomposition $(T, \sigma)$, that has one bag for
  every path between the root $r$ and each leaf $\ell$.
  Then, $T$ forms a path.  As every bag consists of a
  single clique on $h$ vertices, there is a
  clique partition of this tree decomposition with weighted
  width $\log(h+1) \in O(\log \log n)$.
\end{proof}

Finally, we show the algorithmic usefulness of
clique-partitioned treewidth in the following lemma, which
is an extension of the one proposed by de Berg et
al.~\cite{deBerg_framework_intersection}.

\begin{lemma}
  \label{lem:fpt-independent-set}
  Let $G$ be a graph with a clique-partitioned tree
  decomposition $(T, σ)$ of weight~$\tau$.  Then a smallest
  independent set of $G$ can be found in
  $O(2^{\tau} \cdot \mathrm{poly}(n))$ time.
\end{lemma}
\begin{proof}
  We use a standard dynamic programming approach on tree
  decompositions based on \emph{introduce}, \emph{forget},
  and \emph{join} nodes (see for example Cygan
  et. al~\cite{Cygan2015}).  For each node $t \in V(T)$, we
  store a number of partial solutions for the subgraph of
  $G$ induced by the bags of nodes in the subtree below
  $t$.

  A partial solution consists of a subset of the vertices
  in the current bag as well as the size of the total
  partial independent set for the subgraph induced by the
  subtree below the current bag.  This makes it easy to
  initialize partial solutions for leaf nodes in the
  tree decomposition.

  In an introduce node, two new partial solutions are
  created, one where the new vertex is in the independent
  set and one where it is not.  In a forget node, the
  removed vertex is removed from each partial solution.  In
  a join node, the partial solutions from the child-nodes
  are combined by taking their union.

  In a traditional tree decomposition of width $k$, this
  leads to at most $2^{k}$ partial solutions per bag.
  In a clique-partitioned tree decomposition, this
  is even smaller, as there are only $k+1$ ways an
  independent set can intersect a clique of size $k$.
  Thus, assuming $\{\mathcal{P}_{t} \mid t \in V(T)\}$
  denotes the clique partition of weight $\tau$, the number
  of partial solutions that need to be considered per bag
  $t$ are at most
  \[
    \prod_{C \in \mathcal{P}_{t}} (|C| + 1) = 2^{\sum_{C \in \mathcal{P}_{t}} \log(|C| + 1)} = 2^{\tau}.
  \]
  As the number of bags and time spent per bag is
  polynomial, this concludes the proof.
\end{proof}

By the above argumentation, it follows that the
clique-partitioned treewidth introduced in this paper is
upper bounded by the version of de Berg et al. and can be
exponentially lower.  Additionally, it retains some power
in solving NP-hard problems in FPT-time.

\section{The weighted clique partition problem}
\label{sec:clique-partition}

We split the task of computing a clique-partitioned tree
decomposition in two phases.  First, we compute a tree
decomposition, minimizing the traditional tree width.
Secondly, fixing the structure and bags of this
decomposition, we compute a clique partition for every bag.
We note that we already lose optimality by this separation,
i.e., the result may be suboptimal even if we get optimal
solutions in each of the two phases.  However, we expect
that small bags should also allow for low-weight clique
partitions.

In the first phase, we use established algorithms for the
computation of tree decompositions.  Consequently, we focus
on the second step in this section.  To this end, we define
the \textsc{Weighted Clique Partition} problem, short
\textsc{Clique Partition}.  For a given graph $G$ and an
integer $w$, decide if there is a partition of $V(G)$ into
cliques $P_{1}, \dots, P_{k}$ such that
$\prod_{i \in [k]} (|P_{i}|+1) \le w$.  Note that this
function differs from the one in the definition of
clique-partitioned treewidth, but is equivalent, as
$\sum_{i \in [k]} \log(|P_i| + 1) = \log( \prod_{1\le i \le k} (|P_i|+1) )$
and the logarithm is monotonic.

In the following, we prove some technical lemmas that are
useful throughout the section, before showing that
\textsc{Weighted Clique Partition} is NP-complete
(Section~\ref{sec:clique-partition:hardness}).  Afterwards,
we give different heuristic approaches
(Section~\ref{sec:clique-partition:heuristic_approaches})
and an optimal branch-and-bound algorithm in
(Section~\ref{sec:clique-partition:exact}).  We start with
following lemma, which intuitively states that the weight
of a partition is smaller the more imbalanced the
individual weights are, i.e., moving a vertex from a
smaller to a larger clique reduces the total weight.

 \begin{lemma}
  \label{lem:weights_distribute}
  Let $a,b,c,d ∈ ℕ_{0}$ such that $a+b = c+d$ and $a \ge b$, $c \ge d$,
  $d > b$.  Then $(a+1)(b+1) < (c+1)(d+1)$.
\end{lemma}
\begin{proof}
  There is an $x>0$ such that $c = a-x$ and $d = b + x$.
  As $c \ge d$, $x$ can be at most $(a-b)/2$.
  We derive
  \begin{align*}
    (c+1)(d+1) &= (a-x+1)(b+x+1)\\
    &= ab - bx + b + ax - x^{2} + x + a - x + 1\\
    &= (ab + a + b + 1) + ax - bx - x^{2}\\
    &= (a+1)(b+1) + x( a - b - x ).
  \end{align*}
  We have $x(a - b - x) > 0$, as $0 < x \le \frac{a-b}{2}$ and thus
  the claimed strict inequality follows.
\end{proof}

With the above lemma (i.e., repeated applications thereof)
we can compare the weight of two partitions.

\begin{lemma}\label{lem:distribute_repeat}
  Let $\langle s_1, \dots, s_k \rangle$ and
  $\langle r_1, \dots, r_\ell \rangle$ be different
  non-increasing sequences of natural numbers such that
  $2 \le k \le \ell$,
  $\sum_{i \in [k]} s_i = \sum_{i \in [\ell]} r_i$, and
  $s_i \ge r_i$ for all $i \in [k - 1]$.  Then
  $\prod_{i \in [k]} (s_i + 1) < \prod_{i \in [\ell]} (r_i
  + 1)$.
\end{lemma}
\begin{proof}
  This follows from repeatedly applying
  Lemma~\ref{lem:weights_distribute} to go from
  $R = \langle r_1, \dots, r_\ell \rangle$ to
  $S = \langle s_1, \dots, s_k \rangle$ while reducing the
  product in each step.  To make this precise let $i$ be
  the first index where $s_i > r_i$.  We adjust $R$ by
  adding $1$ to $r_i$ ans subtracting $1$ from $r_\ell$.
  Note that this maintains the sum.  We apply
  Lemma~\ref{lem:weights_distribute} with $a = r_i + 1$,
  $b = r_\ell -1$, $c = r_i$, and $d = r_\ell$.  Then, we
  have $(a + 1)(b + 1) < (c + 1)(d + 1)$, i.e., the product
  of the adjusted sequence is smaller than that of the
  original sequence $R$.  Moreover, after a finite number
  of steps, we reach $S$ and thus the product for $S$ is
  smaller than the product for $R$.
\end{proof}

\subsection{Hardness}
\label{sec:clique-partition:hardness}

To prove that \textsc{Weighted Clique Partition} is
NP-complete, we perform a reduction in two steps.  We start
with the NP-hard problem \textsc{3-Coloring}.  It asks for
a given graph whether each vertex can be colored with one
of three colors such that no two neighbors have the same
color.  As an intermediate problem in the reduction, we
introduce \textsc{Weighted Independent Set Partition}.  It
is defined equivalently to \textsc{Weighted Clique
  Partition}, but instead of partitioning the graph into
cliques, we partition it into independent sets, i.e., sets
of pairwise non-adjacent vertices.  Note that independent
sets are cliques in the complement graph and vice versa.
Thus, \textsc{Weighted Independent Set Partition} and
\textsc{Weighted Clique Partition} are computationally
equivalent.  Thus, to obtain the following theorem, it
remains to reduce \textsc{3-Coloring} to \textsc{Weighted
  Independent Set Partition}.

\begin{theorem}\label{thm:np-complete}
  \textsc{Weighted Clique Partition} is NP-complete.
\end{theorem}
\begin{proof}
  Membership in NP is easy to see as polynomial time
  verification of a solution is straightforward.  For
  hardness, we reduce from \textsc{3-Coloring} to
  \textsc{Weighted Independent Set Partition}.  Thus, we
  now assume that we are given a graph $G$ and need to
  transform it into a graph $G'$ and integer $w$ such that
  such that $G$ can be colored with three colors if and
  only if $G'$ has a partition into independent sets of
  weight at most $w$.  We construct $G'$ as follows.  For
  every vertex $v$ of $G$, we add two new vertices $v_{1}$
  and $v_{2}$ that form a triangle together with $v$, but
  have no other edges.  We denote $n = |V(G)|$ and set
  $w = (n + 1)^3$.  Note that any independent set in $G'$
  can contain at most $n$ vertices, because every appended
  triangle admits only one independent vertex.
  
  Assume that $G$ admits a proper three-coloring.  This
  coloring directly translates to a three-coloring of $G'$
  as follows.  Every vertex $v$ of $G$ keeps its color in
  $G'$.  moreover, $v_1$ and $v_2$ each get one of the two
  other colors.  Thus, the coloring classes in $G'$ have
  size exactly $n$ each and form an independent set
  partition with weight $(n+1)^{3}$.

  If otherwise $G$ does not admit a proper three-coloring,
  then neither does $G'$ and there is no partition of $G'$
  into at most three independent sets.  Any partition of
  $V(G')$ into more than three independent sets has a
  weight larger than $(n+1)^{3}$ by
  \cref{lem:distribute_repeat}, as no independent set in
  $G'$ can have more than $n$ vertices.  Consequently, $G$
  is three-colorable if and only if there is a partition of
  $V(G')$ into independent sets with weight at most
  $(n+1)^{3}$.
\end{proof}

\subsection{Heuristic approaches}
\label{sec:clique-partition:heuristic_approaches}
We now explain different approaches to solving the
optimization variant of \textsc{Weighted Clique Partition}
both optimally and heuristically.

Throughout this section we make use of the fact that
enumerating all maximal cliques of a graph is not only
output polynomial~\cite{johnson_enumerating_cliques}, but
also highly feasible in practice as shown by Eppstein,
Löffler, and Strash~\cite{eppstein_cliques}.  We use an
implementation of their algorithm from the
\texttt{igraph}\footnote{\url{https://igraph.org/}}
library.

\subparagraph*{Maximal clique heuristic.}%
Recall from \cref{lem:weights_distribute} that the weight
function favors imbalanced clique sizes over more balanced
ones.  It therefore makes sense to try to find few large
cliques that cover all vertices.  A basic greedy heuristic
that tries to achieve this works as follows.  First, we
enumerate all maximal cliques $\mathcal{C}$ of the graph.
Then we iteratively add one clique to the partition by
greedily selecting the clique with the largest number of
remaining uncovered vertices.  We call this the
\emph{maximal clique heuristic}.

In order to efficiently implement this heuristic, we use a
priority queue to fetch the largest clique and keep track
of the cliques $\mathcal{C}_v \subseteq \mathcal{C}$ that a
vertex $v$ is part of.  This way, after choosing the
remaining vertices of a clique $C \in \mathcal{C}$ as a
partition, we have to update the sizes of
$\mathcal{O}(\sum_{v \in C} |\mathcal{C}_v|)$ cliques.  The
total number of such updates throughout the whole algorithm
is at most the sum of clique sizes in $\mathcal{C}$.  Thus,
using a Fibonacci Heap, a total running time of
$\mathcal{O}(|V| \log |\mathcal{C}| + \sum_{C\in
  \mathcal{C}} |C|)$ can be achieved.  In our
implementation we use a binary heap due to it being
faster in practice.  This costs an additional factor of
$\log |\mathcal C|$ for the second term.

\subparagraph*{Repeated maximal clique heuristic.}%
Note that the MC heuristic does not recompute the maximal
cliques of the remaining graph after selecting a clique.
As deleting the vertices of one clique can have the effect
that a non-maximal clique becomes maximal, the MC heuristic
might miss a clique we would want to select.  The
\emph{repeated maximal clique heuristic} recomputes
the set of maximal cliques after each decision, i.e., it
selects a maximum clique of the remaining graph in each
step.

\subparagraph*{Set Cover heuristics.}%
Observe that for the \textsc{Weighted Clique Partition}
problem, we have to choose a set of cliques of minimum
weight that cover all vertices.  Thus, we essentially have
to solve a weighted \textsc{Set Cover} problem.  As there
are reasonably efficient solvers for \textsc{Set Cover} (or
the equivalent \textsc{Hitting Set} problem), it seems like
a promising approach to use those.  However, this has the
disadvantage, that we would need to list all cliques and
not only the maximal cliques.  Nonetheless, it seems like a
good heuristic to just consider maximal cliques and find a
minimum set cover (unweighted or weighted).

The heuristic consists of two steps.  First, we compute a
minimum set cover, using the maximum cliques as sets and
the vertices as elements.  We consider two variants for
this steps; weighted (a set of size $k$ has weight
$\log(k + 1)$) and unweighted (each set has weight~$1$).
Afterwards, in the second step, we convert the cover into a
partition by assigning the overlap between selected cliques
to only one clique.  We call the resulting two approaches
the \emph{(maximal clique) set cover} and \emph{(maximal
  clique) weighted set cover} heuristics.

For the first step, i.e., solving \textsc{Set Cover}, we
use a state of the art branch-and-bound
solver~\cite{blaesius_hitting_set} for the unweighted case.
Additionally, for the weighted case, we use the
straight-forward formulation of set cover as an ILP and
solve it with Gurobi~\cite{gurobi}.  To the best of our
knowledge, ILP solvers are currently the state-of-the-art
for weighted set cover.

For the second step, we have to compute clique partitions
from the resulting set covers by assigning each vertex that
is covered by multiple cliques to a single one of these
cliques.  The goal is to minimize the weight of the
resulting cliques, i.e., by
Lemma~\ref{lem:weights_distribute}, we want to distribute
them as unevenly as possible.  We employ a simple greedy
heuristic, assigning each vertex to the largest clique it
is part of and braking ties arbitrarily in case of
ambiguity.

At a first glance it seems possible that doing both steps
optimally (solving set cover and resolving the overlaps)
could yield an overall optimal solution.  However, this is
not the case, as briefly discussed in
Appendix~\ref{sec:additional-stuff}.

\subsection{Exact branch-and-bound solver}%
\label{sec:clique-partition:exact}

Our branch-and-bound branches on which clique to select
next.  How to branch is described in
Section~\ref{sec:branching} where we show that we can, in
each step, select a maximal clique and that the cliques of
the resulting sequence are non-increasing in size.  In
Section~\ref{sec:size-lower-bound} and
Section~\ref{sec:valu-sequ-lower}, we describe lower bounds
for pruning the search space, i.e., if the best solution
found so far is better than the lower bound in the current
branch, we can prune that branch.

\subsubsection{Branching}
\label{sec:branching}

The following structural insight enables us to branch on
the maximal cliques.

\begin{lemma}
  \label{lem:optimum_largest_clique_maximal}
  Let $\mathcal{P}$ be a minimum weight clique partition of
  a graph $G$ and let $C \in \mathcal{P}$ be the largest
  clique of $\mathcal{P}$.  Then $C$ is maximal clique in
  $G$.
\end{lemma}
\begin{proof}
  Assume that $C$ is a non-maximal clique.  That is, there
  is a vertex $v ∈ V(G) \setminus C$ with $C ⊆ N(v)$.  Let
  $C' ∈ \mathcal{P}$ be the clique containing $v$.  We
  construct a clique partition $\mathcal{P}'$ by removing
  $v$ from $C'$ and adding it to $C$.  As $C$ was the
  largest clique in $\mathcal{P}$, via
  \cref{lem:weights_distribute} we have
  $(|C|+2)(|C'|) < (|C| + 1) (|C'| + 1)$, contradicting the
  optimality of $\mathcal{P}$.
\end{proof}

Thus, even though not all cliques of an optimal solution
might be maximal, we at least know that the largest one is.
We can use the decision of which maximal clique to select
as the largest one as the branching decision of our
algorithm.  This way, we can solve the optimization variant
of \textsc{Weighted Clique Partition}, i.e., the algorithm
takes a graph $G$ and finds a minimum weight clique
partitioning.

After a clique $C$ has been selected as the largest one,
the remaining problem is to find a clique partition of
$G[V \setminus C]$ that does not use any clique larger
than $C$.  This means that we can view our algorithm as a
simple recursive subroutine that solves the same problem at
every node of the recursion tree.  As input it gets the
graph $G$ and the cliques $\langle C_1, \dots, C_i \rangle$
that have already been selected by previous recursive
calls.  It then tries to compute an optimal
clique partition of the remaining graph
$G' := G \setminus \bigcup_{j \in [i]} C_j$.  This is done
by either returning a trivial solution if $G'$ can be
covered with a single clique or by branching on the
decision of which maximal clique to select as the largest
one for the partition of $G'$.  Note that for this
decision, only maximal cliques that are at most as large as
any of the previously selected cliques
$\langle C_1, \dots, C_i \rangle$ need to be considered.
The result of the subroutine call is then the cheapest
solution found in any of the branches.

In order to quickly obtain a good upper bound, we explore
branches corresponding to larger cliques first.  This way,
the first leaf of the search tree constructs the same
solution as the repeated maximal clique heuristic.

\subsubsection{Size lower bound}
\label{sec:size-lower-bound}

We call the lower bound given by the following lemma the
\emph{size lower bound}.

\begin{lemma}\label{lem:size_lower_bound}
  Let $G$ be a graph with $n$ vertices and
  $\mathcal{P}$ be a clique partition
  of $G$ consisting of cliques of size at most $s$.  Then
  $\mathcal{P}$ has weight at least
  $(s+1)^{\lfloor n / s\rfloor} \cdot ((n \mod s) + 1)$.
\end{lemma}
\begin{proof}
  The stated minimum weight is achieved by a partitioning
  $\mathcal{P}'$ that uses as many cliques of size $s$ as
  possible and one clique with all remaining vertices.  Any
  other partitioning $\mathcal{P}$ using only cliques of
  size at most $s$ is at least as expensive, as it can be
  transformed into $\mathcal{P}'$ by of
  \cref{lem:distribute_repeat}.
\end{proof}

Note that the size lower bound can trivially be evaluated
in constant time.  Even though it is rather basic, we
expect this lower bound to be effective at pruning branches
in which very small cliques are selected early on.

\subsubsection{Valuable sequence lower bound}
\label{sec:valu-sequ-lower}

Note that the size lower bound optimistically assumes that
there are $\lfloor n / s \rfloor$ non-overlapping cliques
of size $s$.  This yields a bad lower bound if, e.g., there
is only one clique of size $s$ while all other cliques are
much smaller.  In the following, we describe an improved
bound based on this observation.  We note that we have to
be careful when considering what clique sizes are available
for the following reason.  Assume the branching has already
picked a clique of size $s$, i.e., subsequent selected
cliques have to have size at most $s$.  Then it seems
natural to derive a lower bound by summing over the sizes
of all maximal cliques of size at most $s$.  However, we
have to account for the fact that selecting (and deleting)
one clique can shrink a maximal clique that was larger than
$s$ to become a clique of size $s$.  Thus, there might me
more cliques of size $s$ available than initially thought.
In order to formalize this, we first introduce a different
problem that considers only sizes of the cliques without
making any assumptions on the overlap between the cliques.

In the \textsc{Valuable Sequence} problem, we are given a
multiset $A$ of natural numbers and a natural number $n$.
The task is to construct a sequence of total value $n$ and
minimum weight.  Such a sequence
$S = \langle s_1, s_2, \dots, s_k \rangle$ consists of
elements $s_i \in A$ such that each number is repeated at
most as often as it appears in $A$.  In the following, we
define value and weight of a sequence and give additional
restrictions to what constitutes a valid sequence.  To this
end, let $S_i = \langle s_1, \dots, s_i \rangle$ for
$i \le k$ denote a prefix of $S$.  We define a value
$\val(s_i)$ for each $s_i$ in the sequence as follows.  The
first element $s_1$ has value $\val(s_1) = s_1$.  For
subsequent elements $s_{i + 1}$, we have
$\val(s_{i + 1}) = \min\{s_{i + 1}, \val(s_i), n -
\val(S_i)\}$, where
$\val(S_i) = \sum_{j \in [i]} \val(s_j)$ is the total value
of the prefix $S_i$.\footnote{Note that $\val(s_{i + 1})$
  only depends on values of previous elements in $S$, i.e.,
  the definition is not cyclic.}  If $\val(s_i) = s_i$, we
say that the element contributes \emph{fully} to the
sequence.  Otherwise, it contributes \emph{partially}.  The
\emph{weight} of $S$ is
$\prod_{i \in [k]} (\val(s_i) + 1)$.  For the subsequence
$S_i$, we call the next element $s_{i+1}$ \emph{eligible}
if $s_{i + 1} - \val(S_i) \le \val(s_i)$; $s_1$ is always
eligible.  The sequence $S$ is \emph{valid} if each element
is eligible.

To make the connection back to \textsc{Weighted Clique
  Partition}, interpret the numbers in $A$ as the clique
sizes.  The total value $n$ corresponds to the number of
vertices that have to be covered.  The value $\val(s_i)$
corresponds to the number of vertices from the maximal
clique of size $s_i$ in $G$ that have not been covered by
previous cliques, i.e., the number of vertices that are
newly covered in step $i$.  Note that in step $i+1$, at
least $s_{i} - \val(S_i)$ new vertices are covered as only
$\val(S_i)$ have been covered previously.  Thus, the
eligibility requirement ensures that the number of vertices
covered in step $i + 1$ is not larger than the number of
vertices covered in step $i$ (recall, that we can assume
the chosen cliques to form a non-increasing sequence).
Moreover, for the definition of $\val(s_{i + 1})$, note
that the minimum with $\val(s_i)$ ensures that the sequence
of values is non-increasing and the minimum with
$n - \val(S_i)$ ensures that the total value is $n$.

The following two lemmas formalize this connection between
\textsc{Valuable Sequence} and \textsc{Weighted Clique
  Partition}.  Afterwards, we discuss how \textsc{Valuable
  Sequence} can be solved optimally.

\begin{lemma}\label{lemma_partition_mapping_injective}
  Let $\mathcal{P}$ be a minimum weight clique partition of
  a graph $G$ and let $\mathcal{C}$ be the set of maximal
  cliques in $G$.  Then, any mapping
  $f: \mathcal{P} \to \mathcal{C}$ with $P ⊆ f(P)$ for
  each $P∈\mathcal{P}$ is injective and there exists at
  least one such mapping.
\end{lemma}
\begin{proof}
  There are mappings from $\mathcal{P}$ to $\mathcal{C}$,
  because each clique $P ∈ \mathcal{P}$ is either a maximal
  clique or a subset of a larger maximal clique.  Assume
  that a mapping $f: \mathcal{P} \mapsto \mathcal{C}$ is
  not injective.  Then, there are two partition classes $P$
  and $P'$ that are mapped to the same clique
  $C ∈ \mathcal{C}$.  Thus, these partition
  classes could be merged, contradicting the
  optimality of $\mathcal{P}$.
\end{proof}

\begin{lemma}
  Let $G$ be a graph with maximal cliques
  $\mathcal{C} = \{C_1, \dots, C_k \}$ and let $(A, n)$
  with $A = \{|C_1|, \dots, |C_k|\}$ and $n = |V(G)|$ be an
  instance of \textsc{Valuable Sequence}.  The weight of a
  minimum solution of $(A, n)$ is a lower bound for the
  weight of every clique partition of $G$.
\end{lemma}
\begin{proof}
  We now show that for a minimum clique partition
  $\mathcal{P}$ of $G$, we find a solution $S$ of $(A, n)$
  whose weight is at most the weight of $\mathcal{P}$.

  Let $P_1, \dots, P_{k'}$ be the cliques of $\mathcal{P}$
  sorted by size in decreasing order.  We can think of
  $\mathcal{P}$ as constructed iteratively in that order,
  so that each $P_i$ is a maximal clique in
  $G[V\setminus (\bigcup_{j \in [i - 1]} P_j)]$.

  We construct $S$ iteratively until $\val(S) = n$.  For
  each element $s_i$ in the sequence, we prove by induction
  that $\val(s_i) \ge |P_{i}|$ except for the last element.
  This then lets us use \cref{lem:distribute_repeat} to
  obtain that the weight of $S$ is at most the weight of
  $\mathcal P$.  For $i=1$, we simply choose
  $s_1 = |P_1| ∈ A$.  Since the first element
  always contributes fully, we have $\val(s_1) = |P_{1}|$.

  Assuming we constructed the sequence until $i$, we
  continue with step $i + 1$ as follows.  If $P_{i + 1}$ is
  one of the initial maximal cliques in $\mathcal{C}$, then
  we can simply choose $s_{i + 1} = |P_{i + 1}| \in A$.
  Note that $s_{i + 1}$ is eligible, as
  $s_{i + 1} = |P_{i + 1}| \le |P_i| \le \val(s_i)$, which
  in particular implies
  $s_{i + 1} - \val(S_i) \le \val(s_i)$.  In this case,
  $s_{i + 1}$ contributes fully, i.e.,
  $\val(s_{i + 1}) = |P_{i + 1}|$, which implies the claim.
  
  Otherwise, if $P_{i + 1}$ is not in $\mathcal{C}$, it is
  at least a subset of some clique $C \in \mathcal{C}$ such
  that $P_{i + 1} = C \setminus \bigcup_{j \in [i]} P_j$.
  We choose $s_{i + 1} = |C|$.  The eligibility of
  $s_{i + 1}$ follows from the facts that the cliques in
  $\mathcal P$ are ordered non-increasingly, i.e.
  $|P_{i}| \ge |P_{i + 1}|$, and that
  $\val(s_j) \ge |P_{j}|$ holds by induction for all
  $j < i + 1$:
  \begin{equation*}
    \val(s_i)
    \ge |P_{i}|
    \ge |P_{i + 1}|
    = \Big|C \setminus \bigcup_{j \in [i]} P_j\Big| 
    \ge |C| - \sum_{j \in [i]} |P_j| 
    \ge s_{i + 1} - \sum_{j \in [i]} \val(s_j)
    = s_{i + 1} - \val(S_i).
  \end{equation*}
  Note that $s_{i + 1}$ contributes partially, i.e.,
  $\val(s_{i + 1}) = \val(s_{i})$ unless this is the last
  item in $S$.  As we just argued, we have
  $\val(s_i) \ge |P_{i + 1}|$ and thus
  $\val(s_{i + 1}) \ge |P_{i + 1}|$, proving the claim.

  To conclude, observe that our construction of $S$
  implicitly defines a mapping from $\mathcal{P}$ to
  $\mathcal{C}$ as in
  \cref{lemma_partition_mapping_injective}.  As such a
  mapping is injective, no number in $A$ is chosen twice.
  Moreover as we have $\val(s_i) \ge |P_i|$ for $i < k$,
  but both sum to $n$, the weight of $\mathcal{P}$ is at
  least the weight of $S$ by~\cref{lem:distribute_repeat}.
\end{proof}

\textsc{Valuable Sequence} can be solved optimally with a
simple greedy algorithm.  We call the resulting lower bound
the \emph{valuable sequence} bound.

\begin{theorem}
  \label{thm:valu_sequ_greedy}
  An instance $(A, n)$ of \textsc{Valuable Sequence} can be
  solved in $O(|A| + n)$ time.
\end{theorem}
\begin{proof}
  We construct a solution $S = s_1, \dots, s_i$ by
  iteratively choosing an eligible and not yet chosen
  number $a \in A$ with maximum value, until the value of
  the sum reaches $n$.

  We note that this greedy strategy maximizes how many
  numbers in $A$ are eligible, as the corresponding upper
  bound $\val(S) + \val(s_i)$ decreases as slowly as
  possible.  The optimality of the produced sequence
  $S = s_1,\dots,s_i$ follows, again, via
  \cref{lem:distribute_repeat} as for $j \in [i]$, the
  value $\val(s_{j})$ is at least as large as the value of
  any other number that can be chosen in round $j$.

  Regarding the running time, the greedy strategy can be
  implemented by sorting the numbers in $A$ (in $O(|A|+n)$
  time) and keeping track of the largest unchosen number
  that is eligible and contributes fully, as well as the
  smallest unchosen number that can contribute partially
  (which is larger than the ones that can contribute
  fully).  Both of these values can be updated in constant
  time each time a number has been chosen.
\end{proof}

\subsubsection{Sufficient weight reduction}
\label{sec:suff-weight-redu}

To speed-up the computation of clique partitions for all
bags of a tree decomposition, we additionally apply the
following reduction rule.  In the \emph{sufficient weight}
reduction, we immediately accept the first solution that is
lighter or equally light as the largest weight of any of
the already considered bags.

\section{Evaluation}
\label{sec:eval}

With our evaluation, we aim to answer the following
questions.

\begin{enumerate}
  \item How do the different algorithms compare in regards
    to run time and quality?
  \item How do the algorithms scale?
  \item What is the impact of the lower bounds and the
    reduction rules on the performance of the exact
    branch-and-bound solver?
  \item How do different network properties influence the
    performance of the algorithms?
  \item How do the resulting upper bounds on the
    clique-partitioned treewidth compare to traditional
    treewidth?
\end{enumerate}

\subparagraph*{Experimental setup.}
\label{sec:eval:solvers:setup}

Our implementation is written in Python. The source code
along with all evaluation scripts and results is available
on our public GitHub
repository\footnote{\url{https://github.com/marcwil/cptw_code}}.
The experiments were run with Python 3.10.1 on a Gigabyte
R282-Z93 (rev. 100) server (2250MHz) with 1024GB DDR4
(3200MHz) memory.

For each input graph, we perform the following two steps.
First, we compute a tree decomposition using the heuristics
implemented in the HTD library~\cite{htd}.  Specifically,
we use the min-fill-in heuristic, which is known to provide
a good tradeoff between run time and solution
quality~\cite{DBLP:conf/icdt/ManiuSJ19}.  Secondly, we
solve the \textsc{Weighted Clique Partition} problem for
each bag of the tree decomposition using all algorithms
proposed in \cref{sec:clique-partition}.

We use a time limit of five minutes for the heuristic
computation of low-weight tree decompositions with the HTD
library. For the \textsc{Weighted Clique Partition} algorithms, we set a time
limit of three minutes per bag and five minutes in total.

To discern the different solvers from
\cref{sec:clique-partition:heuristic_approaches,sec:clique-partition:exact},
in our plots, we use the following abbreviations: branch
and bound solver (B\&B), maximal clique set cover heuristic
(SC), maximal clique weighted set cover heuristic (WSC),
maximal clique heuristic (MC), and repeated maximal clique
heuristic (RMC).

\subparagraph{Input instances.}

For the input, we use a large collection of real-world
networks as well as generated networks. For the latter, we
use geometric inhomogeneous random graphs
(GIRGs)~\cite{bringmann_girgs}, which resemble real-world
networks in regards to important properties and have been
shown to be well suited for the evaluation of
algorithms~\cite{blasius_external_validity}. GIRGs can be
generated efficiently~\cite{weyand_girgs} and allow to vary
the power-law exponent ($\mathrm{ple}$) of the degree
distribution controlling its heterogeneity, as well as a
parameter $\alpha$ controlling the locality by either
strengthening the influence of the geometry (high values of
$α$) or increasing the probability for random edges not
based on the geometry (low values of $α$). We mainly use
the following two datasets, where each graph has been
reduced to its largest connected component.
\begin{itemize}
  \item A collection of $\num{2967}$ real-world
        networks~\cite{blasius_3006_networks} that
        essentially consists of all networks with at most
        \SI{1}{M} edges from Network
        Repository~\cite{DBLP:conf/aaai/RossiA15};
        see~\cite{blasius_external_validity} for details.
  \item GIRGs with $n \in \{500, 5000, 50000\}$ vertices,
        expected average degree $10$, dimension $1$,
        $\mathrm{ple} \in \{2.1, 2.3, 2.5, 2.7, 2.9\}$, and
        $\alpha \in \{1.25, 2.5, 5, \infty\}$.  For each
        parameter configuration, we generate ten networks
        with different random seeds, to smooth out random
        variations.
\end{itemize}

\subsection{Performance comparison}
\label{sec:perf-comp}

Here, we evaluate the performance of our \textsc{Clique Partition} approaches on the two datasets.

\subparagraph*{Generated instances.}
\label{sec:eval:solvers:compare_relative}

\begin{figure}[]
  \centering
  \begin{subfigure}[b]{0.31\textwidth}
    \includegraphics[width=\textwidth]{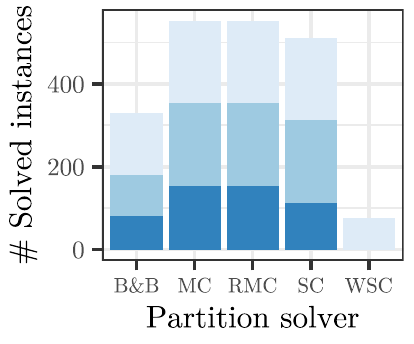}
    \caption{Number of solved instances with 500 (bright)
      \SI{5}{k} (medium) and \SI{50}{k} (dark) vertices.}
    \label{fig:compare_partition_solvers:timelimit}
  \end{subfigure}
  ~ 
  \begin{subfigure}[b]{0.31\textwidth}
    \includegraphics[width=\textwidth]{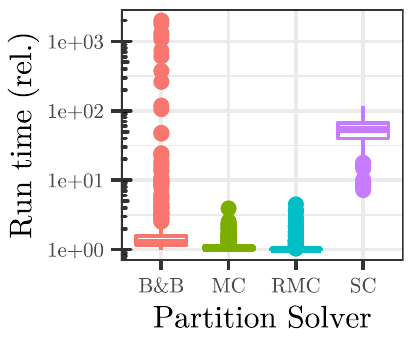}
    \caption{Distribution of run time relative to fastest
      solver within time limit on a given graph.}
    \label{fig:compare_partition_solvers:time}
  \end{subfigure}
  ~ 
  \begin{subfigure}[b]{0.31\textwidth}
    \includegraphics[width=\textwidth]{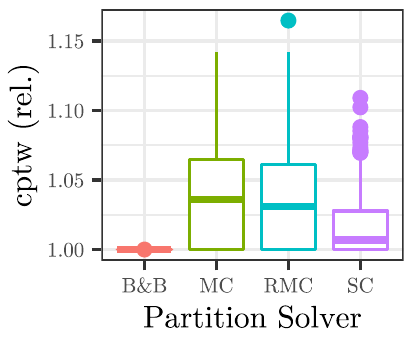}
    \caption{Distribution of obtained width relative to
      best found solution on a given graph.}
    \label{fig:compare_partition_solvers:quality}
  \end{subfigure}
  \caption{Comparison of run time and solution quality of
    the different exact (red), greedy (green, blue) and set cover
    based (violet) solvers for the \textsc{Weighted Clique Partition} problem
    on GIRGs.}%
  \label{fig:compare_partition_solvers}
\end{figure}

In \cref{fig:compare_partition_solvers}, we compare the run
times as well as the solution quality of the different
considered \textsc{Clique Partition} algorithms on the dataset of
generated networks.  In
\cref{fig:compare_partition_solvers:timelimit}, we show how
many of the 600 instances were solved within the time limit
by each solver.  While the greedy heuristics are able to
finish on almost all instances, the set cover heuristic and
the branch-and-bound solver get timed on some of the larger
networks with \SI{5}{k} and \SI{50}{k} vertices.  The
weighted set cover heuristic performs much worse, finishing
only on few networks.  We therefore exclude it from the
other comparisons.

In
\cref{fig:compare_partition_solvers:time,fig:compare_partition_solvers:quality},
we compare the performance for all instances that were
solved within the time limit by all other algorithms.
\cref{fig:compare_partition_solvers:time} shows the run
time of each algorithm relative to the fastest one on each
instance.  \cref{fig:compare_partition_solvers:quality}
shows the obtained upper bound on the cp-treewidth relative
to the optimal solution computed by the branch-and-bound
solver.

Our findings are as follows.  The branch-and-bound solver
solves the fewest instances of the four considered
algorithms, but is quick on most of the instances it is
able to solve within the time limit.  Both greedy
heuristics (MC and RMC) are similarly fast, significantly
outcompeting the other approaches.  In terms of quality,
all three heuristics perform well, achieving solutions
within few percent of the optimum.  The set cover heuristic
slightly outperforms the greedy heuristics in terms of
quality, but pays for this with substantially higher running
time.

\subparagraph*{Real-world networks.}
\begin{figure}[]
  \begin{minipage}[t][][b]{.45\linewidth}
    \centering
    \includegraphics[]{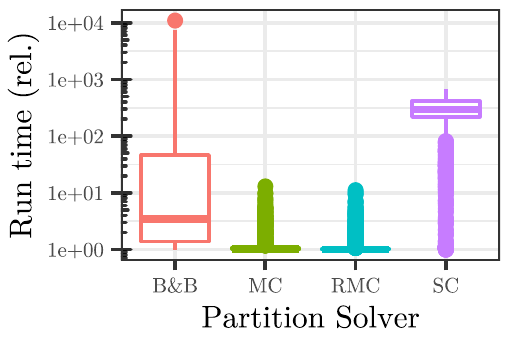}
    \caption
      {%
        Distribution of run time relative to fastest
      solver within time limit on our set of real-world networks.%
        \label{fig:compare_partition_solvers_rw:time}%
      }%
  \end{minipage}\hfill
  \begin{minipage}[t][][t]{.5\linewidth}
    \centering
    \captionof{table}
      {%
        Distribution of obtained cp-treewidth relative
    to optimum clique partition on our set of real-world networks.%
        \label{tbl_rw_rel_weight}%
      }
    \begin{tabular}{rllll}
      \hline
      \textbf{Measure} & \textbf{MC} & \textbf{RMC} & \textbf{SC} \\
      \hline
      Mean & 1.008 & 1.009 & 1.002 \\
      Median & 1 & 1 & 1 \\
      90th percentile & 1.035 & 1.036 & 1.000 \\
      99th percentile & 1.108 & 1.121 & 1.046 \\
      Maximum & 1.254 & 1.192 & 1.113 \\
      \hline
    \end{tabular}
  \end{minipage}
\end{figure}

We complement the above evaluation of our \textsc{Clique Partition}
algorithms, by comparing their performance on the collection
of real-world networks.  As above, we exclude the weighted
set cover heuristic.  The other four approaches were able
to finish on 1243 (B\&B), 2619 (MC), 2622 (RMC), and 2204
(SC) of the 2967 networks within the time limit.  We
compare our algorithms on the 1237 networks that were
solved by all four approaches.
\cref{fig:compare_partition_solvers_rw:time} shows the run
time of each solver relative to the fastest solver on each
instance.  In \cref{tbl_rw_rel_weight} we describe the
distribution of the obtained upper bounds on the
cp-treewidth relative to the optimal solution found by the
branch-and-bound solver.

Our results are the following.  In general, our
observations on generated networks are replicated on the
real-world networks.  Even though the branch-and-bound
algorithm solved fewer instances than the set cover
heuristic, it is comparatively faster on the networks it is
able to solve.  Both approaches are, however, considerably
slower than the greedy heuristics and this difference is
more pronounced than on the generated networks.  Regarding
the solution quality, all three heuristic solvers perform
even better than on the generated networks, with only a
tiny fraction of instances not being solved almost
optimally.

\subparagraph*{Discussion.} We find that the
proposed algorithms show good performance both on generated
and real-world instances.  Although, the branch-and-bound
solver was only able to solve about half of the considered
networks, it's run time typically beats the set cover
heuristic on the networks it can solve.  In addition, it is
a valuable tool for evaluating the solution quality of the
other approaches.  We find that especially the set cover
heuristic, but also the greedy heuristics (MC and RMC)
often find close to optimal clique partitions.  Due to
their excellent trade-off between speed and solution
quality, the greedy heuristics are probably the best
approach in most practical settings.  In general, we do not
expect that there is substantial room for improvement in
the engineering of \textsc{Clique Partition} solvers for the
computation of cp-treewidth.  Instead, in order to achieve
better upper bounds, we suggest future research to optimize
the tree decomposition and the partition into cliques at
the same time.

\subsection{Run time scaling}
\label{sec:run-time-scaling}

\begin{figure}
  \centering
  \includegraphics[]{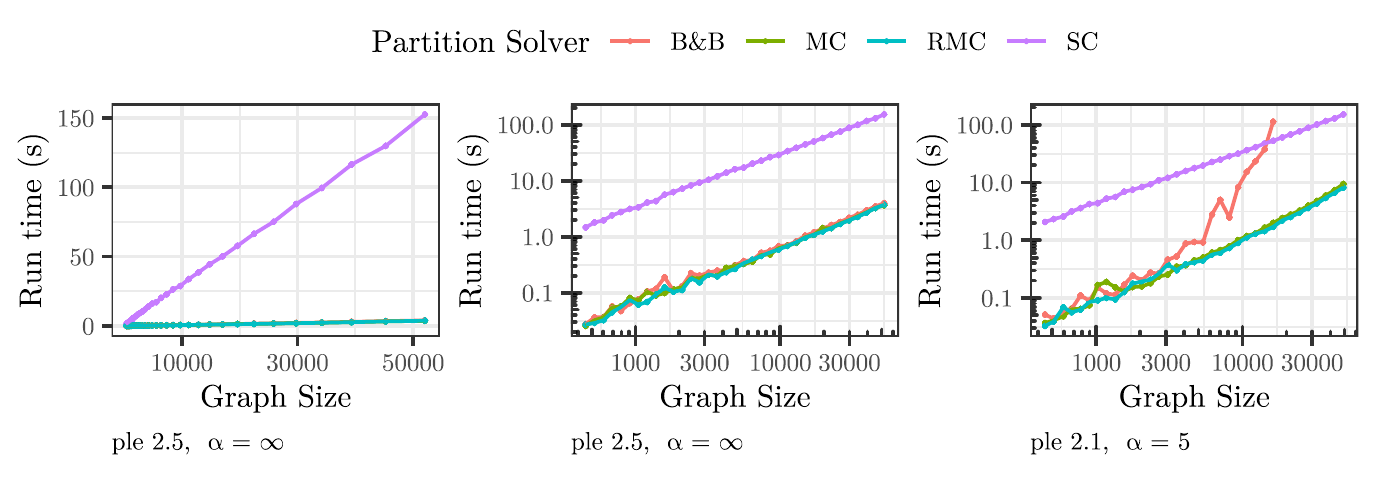}
  \caption{Scaling behavior of \textsc{Clique Partition}
    algorithms on GIRGs with different parameters.}
  \label{fig:scaling_comparison}
\end{figure}

Next, we consider the scaling behavior of our solvers.  For
this, we generated GIRGs of varying sizes up to around
\SI{50}{k} vertices for various parameters.  As in
\cref{sec:perf-comp}, we did not evaluate the weighted set
cover heuristic.  \cref{fig:scaling_comparison} shows the
run times for GIRGs with two different parameter
configurations.  On the networks with high locality
($α = \infty$), all four approaches seem to have close to
linear run time, despite enumerating all maximal cliques
present in each bag of the tree decomposition.  However, as
we decrease the locality ($α = 5$) the performance of the
branch-and-bound solver deteriorates while the greedy
heuristics and especially the set cover heuristic are only
slightly affected.  In the logarithmic plot, we observe
clearly super-polynomial scaling behavior only for the
branch-and-bound solver.  Further experiments on a larger
grid of parameter settings confirm the above findings.

\subsection{Branch-and-bound: lower bounds and reduction
  rule}
\label{sec:eval:solvers:reduction_rules}

\begin{figure}[]
  \centering
  \includegraphics[]{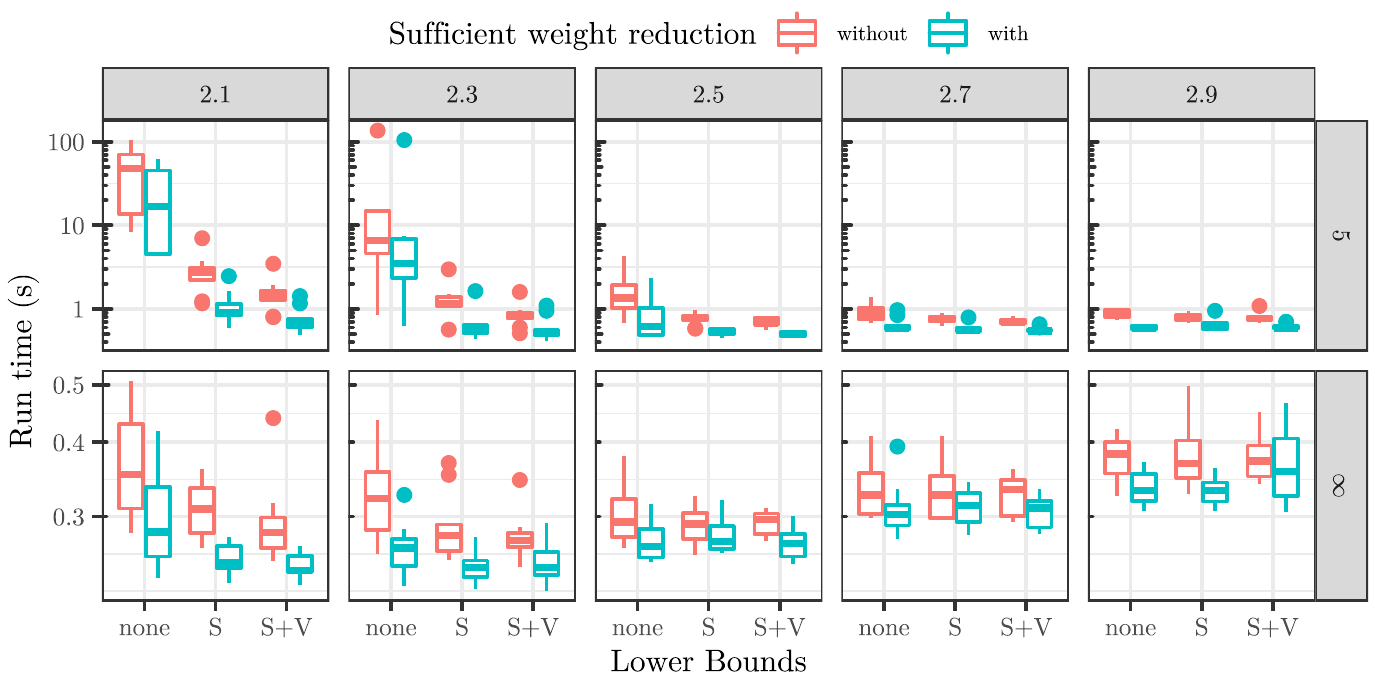}
  \caption{Run time of different variants of the
    branch-and-bound solver on GIRGs with \SI{5}{k}
    vertices and different values for the power-law
    exponent (left to right) and $α$ (top / bottom).}
  \label{fig:compare_redrules_box_5000}
\end{figure}

In the following, we evaluate the effectiveness of the
lower bounds and the reduction rule in speeding up our
branch-and-bound solver.  For this, we use the dataset of
generated networks.  As the performance without lower
bounds does not allow for the timely evaluation on larger
instances, we consider only graphs generated with \SI{5}{k}
vertices.  \cref{fig:compare_redrules_box_5000} shows the
average run time without lower bounds (none), with only the
size lower bound (S) and with the valuable sequence bound
in addition to the size bound (S+V) as well as with and
without the sufficient weight reduction for different
network parameters.  We only show $α ∈ \{5, ∞\}$, as for
lower values the variant without lower bounds did not
finish within the time limit.

We find that especially for smaller power-law exponents,
the lower bounds bring large speed-ups of up to multiple
orders of magnitude.  The additional gain of using the size
lower bound is much larger than that of the much simpler
valuable sequence bound.  The sufficient weight reduction
yields similar speed-ups for all settings.  Overall, we
conclude that the lower bounds are effective in speeding up
the branch-and-bound solver.  On a more general note, it is
striking how strongly all variants of the solver are
affected by lower values of $α$, especially also below the
values shown in \cref{fig:compare_redrules_box_5000}.  In
additional experiments we found that the above observations
also apply to the remainder of the dataset, even though for
\SI{50}{k} vertices the time limit is reached even more
frequently.

\subsection{Impact of network properties}
\label{sec:eval:impact-netw-prop}

\begin{figure}
  \centering
  \includegraphics[]{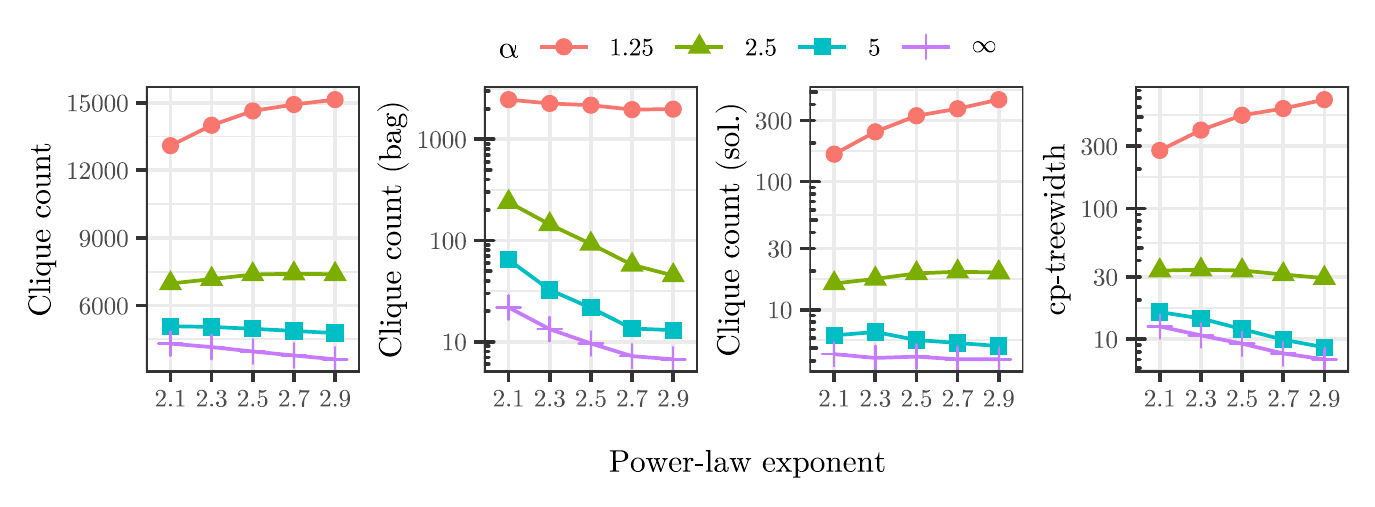}
  \caption{Total clique count (number of maximal cliques)
    per network, and highest clique count in any bag of a
    greedy tree decomposition as well in the lowest weight
    clique partition of any bag, and clique-partitioned
    treewidth (lowest upper bound) of the entire instance
    on GIRGs with $\SI{5}{k}$ vertices and varying
    parameters.  Note the logarithmic y-axes on all except
    the first plot.}
  \label{fig:girgs:clique_properties}
\end{figure}

At multiple points throughout the last sections, we found
that, especially for the branch-and-bound algorithm, the
performance strongly depended on the parameter $α$
controlling the locality of the generated networks.

In order to better understand this, we study the structure
of cliques in the generated networks depending on their
parameters.  Specifically, for each network we count the
number of maximal cliques in the graph, we count the number
of maximal cliques in each bag of the tree decomposition
and take the maximum, we count the number of cliques used
per bag in the clique-partitioned tree decomposition and
take the maximum, and consider the width of the
clique-partitioned tree decomposition.  The
clique-partitioned tree decompositions are obtained using
the MC and MCR heuristic.
\cref{fig:girgs:clique_properties} shows these values for
GIRGs with varying power-law exponent and $α$.

We see that with decreasing values of $α$, all considered
measures increase.  However, while the total number of
maximal cliques in the network only increases by a factor
of roughly 4 to 10, the highest number of cliques
intersecting some bag of the tree decomposition as well as
the highest number of cliques in a lowest-weight clique
partition increase by multiple orders of magnitude.
Intuitively, this can be explained by cliques starting to
fray if the locality is too low.  This explains, why the
\textsc{Clique Partition} problem is harder on GIRGs
with lower values of $α$, which slows down the branch-and-bound
algorithm.  We also observe, that the obtained upper bounds
on the cp-treewidth are not much lower than the highest
number of cliques per bag of a solution, explaining the
good performance of the set cover heuristic.

\subsection{Clique-partitioned treewidth compared to
  traditional treewidth}
\label{sec:eval:param_value}

Here we consider the data set of real-world networks.  As
we have seen in Section~\ref{sec:perf-comp}, the maximal
clique and repeated maximal clique heuristics are efficient
and tend to perform well in terms of quality.  Thus, we use
these two heuristics to find an upper bound on the
clique-partitioned treewidth.

\begin{figure}[]
  \centering
  \begin{subfigure}[b]{0.48\textwidth}
    \includegraphics[]{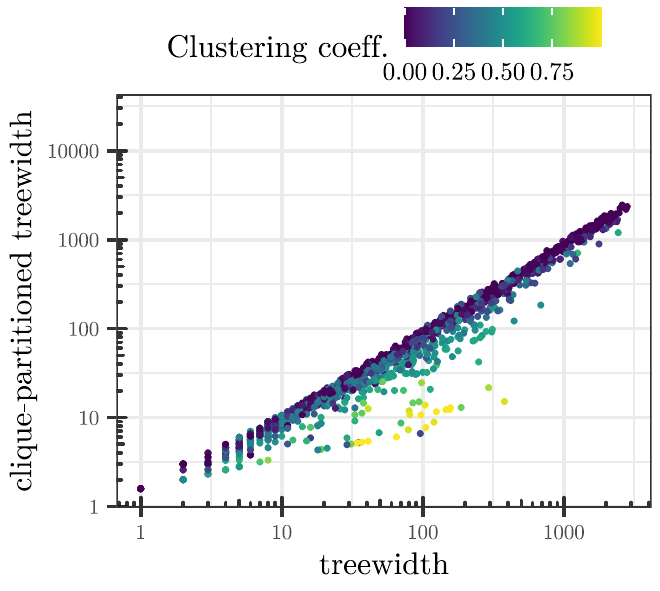}
    \caption{Dependency between clustering coefficient and
      heuristic upper bounds on clique-partitioned treewidth and
      treewidth.}
    \label{fig:rw_ub_wtw_clustering:scatter}
  \end{subfigure}
  ~ 
  \begin{subfigure}[b]{0.48\textwidth}
    \includegraphics[]{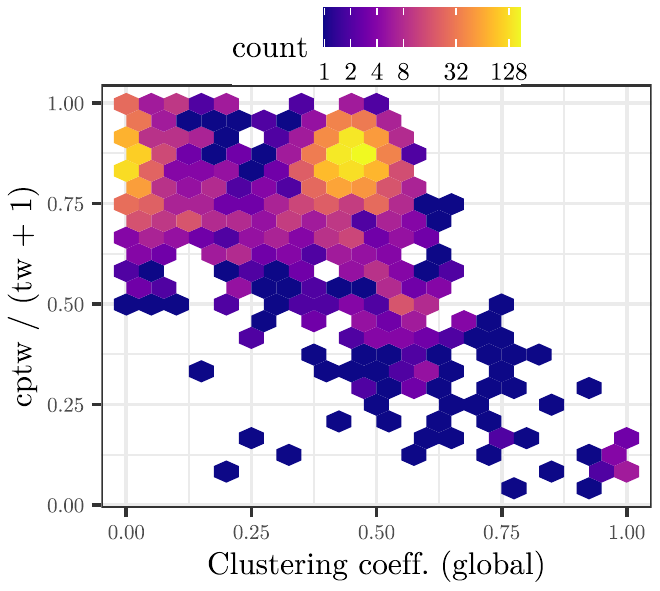}
    \caption{Dependency between clustering coefficient and
      relative difference between clique-partitioned treewidth and
      treewidth.}
    \label{fig:rw_clustering:density}
  \end{subfigure}
  \caption{Upper bounds for clique-partitioned treewidth on
    large real-world networks.
  }%
  \label{fig:rw_ub}
\end{figure}

In \cref{fig:rw_ub_wtw_clustering:scatter} we compare the
obtained upper bounds for the weighted treewidth and the
treewidth.  Even though the parameter does not decrease
much for the majority networks, there are some networks on
which substantial reductions are achieved.  This is
particularly true for networks with high clustering
coefficient, where for some instances our
clique-partitioned tree decomposition has width 10 while
the corresponding traditional tree decomposition has width
above 100.  This correspondence with the clustering
coefficient fits well to the observations in
\cref{sec:eval:impact-netw-prop}. For the networks for
which we do not yet see a big improvement, it would be
interesting to see whether adjusting the computation of the
initial tree decomposition can yield better bounds; see
also the discussion in \cref{sec:perf-comp}.

\bibliography{paper}

\begin{thebibliography}{10}

\bibitem{htd}
Michael Abseher, Nysret Musliu, and Stefan Woltran.
\newblock htd - {A} free, open-source framework for (customized) tree
  decompositions and beyond.
\newblock In Domenico Salvagnin and Michele Lombardi, editors, {\em Integration
  of {AI} and {OR} Techniques in Constraint Programming - 14th International
  Conference, {CPAIOR} 2017, Padua, Italy, June 5-8, 2017, Proceedings}, volume
  10335 of {\em Lecture Notes in Computer Science}, pages 376--386. Springer,
  2017.
\newblock \href {https://doi.org/10.1007/978-3-319-59776-8\_30}
  {\path{doi:10.1007/978-3-319-59776-8\_30}}.

\bibitem{aronis_tree_clique_width}
Chris Aronis.
\newblock The algorithmic complexity of tree-clique width.
\newblock {\em CoRR}, abs/2111.02200, 2021.
\newblock URL: \url{https://arxiv.org/abs/2111.02200}, \href
  {http://arxiv.org/abs/2111.02200} {\path{arXiv:2111.02200}}.

\bibitem{blasius_external_validity}
Thomas Bl{\"{a}}sius and Philipp Fischbeck.
\newblock On the external validity of average-case analyses of graph
  algorithms.
\newblock In Shiri Chechik, Gonzalo Navarro, Eva Rotenberg, and Grzegorz
  Herman, editors, {\em 30th Annual European Symposium on Algorithms, {ESA}
  2022, September 5-9, 2022, Berlin/Potsdam, Germany}, volume 244 of {\em
  LIPIcs}, pages 21:1--21:14. Schloss Dagstuhl - Leibniz-Zentrum f{\"{u}}r
  Informatik, 2022.
\newblock \href {https://doi.org/10.4230/LIPIcs.ESA.2022.21}
  {\path{doi:10.4230/LIPIcs.ESA.2022.21}}.

\bibitem{weyand_girgs}
Thomas Bl{\"{a}}sius, Tobias Friedrich, Maximilian Katzmann, Ulrich Meyer,
  Manuel Penschuck, and Christopher Weyand.
\newblock Efficiently generating geometric inhomogeneous and hyperbolic random
  graphs.
\newblock In Michael~A. Bender, Ola Svensson, and Grzegorz Herman, editors,
  {\em 27th Annual European Symposium on Algorithms, {ESA} 2019, September
  9-11, 2019, Munich/Garching, Germany}, volume 144 of {\em LIPIcs}, pages
  21:1--21:14. Schloss Dagstuhl - Leibniz-Zentrum f{\"{u}}r Informatik, 2019.
\newblock \href {https://doi.org/10.4230/LIPIcs.ESA.2019.21}
  {\path{doi:10.4230/LIPIcs.ESA.2019.21}}.

\bibitem{blaesius_hyp_sep_treewidth}
Thomas Bl{\"{a}}sius, Tobias Friedrich, and Anton Krohmer.
\newblock Hyperbolic random graphs: Separators and treewidth.
\newblock In Piotr Sankowski and Christos~D. Zaroliagis, editors, {\em 24th
  Annual European Symposium on Algorithms, {ESA} 2016, August 22-24, 2016,
  Aarhus, Denmark}, volume~57 of {\em LIPIcs}, pages 15:1--15:16. Schloss
  Dagstuhl - Leibniz-Zentrum f{\"{u}}r Informatik, 2016.
\newblock \href {https://doi.org/10.4230/LIPIcs.ESA.2016.15}
  {\path{doi:10.4230/LIPIcs.ESA.2016.15}}.

\bibitem{blaesius_hitting_set}
Thomas Bl{\"{a}}sius, Tobias Friedrich, David Stangl, and Christopher Weyand.
\newblock An efficient branch-and-bound solver for hitting set.
\newblock In Cynthia~A. Phillips and Bettina Speckmann, editors, {\em
  Proceedings of the Symposium on Algorithm Engineering and Experiments,
  {ALENEX} 2022, Alexandria, VA, USA, January 9-10, 2022}, pages 209--220.
  {SIAM}, 2022.
\newblock \href {https://doi.org/10.1137/1.9781611977042.17}
  {\path{doi:10.1137/1.9781611977042.17}}.

\bibitem{blasius_3006_networks}
Thomas Bläsius and Philipp Fischbeck.
\newblock {3006 Networks (unweighted, undirected, simple, connected) from
  Network Repository}, May 2022.
\newblock \href {https://doi.org/10.5281/zenodo.6586185}
  {\path{doi:10.5281/zenodo.6586185}}.

\bibitem{bringmann_girgs}
Karl Bringmann, Ralph Keusch, and Johannes Lengler.
\newblock Geometric inhomogeneous random graphs.
\newblock {\em Theor. Comput. Sci.}, 760:35--54, 2019.
\newblock \href {https://doi.org/10.1016/j.tcs.2018.08.014}
  {\path{doi:10.1016/j.tcs.2018.08.014}}.

\bibitem{Cygan2015}
Marek Cygan, Fedor~V. Fomin, {\L}ukasz Kowalik, Daniel Lokshtanov, D{\'a}niel
  Marx, Marcin Pilipczuk, Micha{\l} Pilipczuk, and Saket Saurabh.
\newblock {\em Treewidth}, pages 151--244.
\newblock Springer International Publishing, Cham, 2015.
\newblock \href {https://doi.org/10.1007/978-3-319-21275-3_7}
  {\path{doi:10.1007/978-3-319-21275-3_7}}.

\bibitem{dallard_computing_tree_ind_num}
Cl{\'{e}}ment Dallard, Fedor~V. Fomin, Petr~A. Golovach, Tuukka Korhonen, and
  Martin Milanic.
\newblock Computing tree decompositions with small independence number.
\newblock {\em CoRR}, abs/2207.09993, 2022.
\newblock \href {http://arxiv.org/abs/2207.09993} {\path{arXiv:2207.09993}},
  \href {https://doi.org/10.48550/arXiv.2207.09993}
  {\path{doi:10.48550/arXiv.2207.09993}}.

\bibitem{dallard_treewidth_clique_i}
Cl{\'{e}}ment Dallard, Martin Milanic, and Kenny Storgel.
\newblock Treewidth versus clique number. i. graph classes with a forbidden
  structure.
\newblock {\em {SIAM} J. Discret. Math.}, 35(4):2618--2646, 2021.
\newblock \href {https://doi.org/10.1137/20M1352119}
  {\path{doi:10.1137/20M1352119}}.

\bibitem{dallard_treewidth_clique_iii}
Cl{\'{e}}ment Dallard, Martin Milanic, and Kenny Storgel.
\newblock Treewidth versus clique number. {III.} tree-independence number of
  graphs with a forbidden structure.
\newblock {\em CoRR}, abs/2206.15092, 2022.
\newblock \href {http://arxiv.org/abs/2206.15092} {\path{arXiv:2206.15092}},
  \href {https://doi.org/10.48550/arXiv.2206.15092}
  {\path{doi:10.48550/arXiv.2206.15092}}.

\bibitem{dallard_treewidth_clique_ii}
Clément Dallard, Martin Milanič, and Kenny Štorgel.
\newblock Treewidth versus clique number. ii. tree-independence number, 2021.
\newblock URL: \url{https://arxiv.org/abs/2111.04543}, \href
  {https://doi.org/10.48550/ARXIV.2111.04543}
  {\path{doi:10.48550/ARXIV.2111.04543}}.

\bibitem{deBerg_framework_intersection}
Mark de~Berg, Hans~L. Bodlaender, S{\'{a}}ndor Kisfaludi{-}Bak, D{\'{a}}niel
  Marx, and Tom~C. van~der Zanden.
\newblock A framework for eth-tight algorithms and lower bounds in geometric
  intersection graphs.
\newblock In {\em Proceedings of the 50th Annual {ACM} {SIGACT} Symposium on
  Theory of Computing, {STOC} 2018, Los Angeles, CA, USA, June 25-29, 2018},
  pages 574--586. {ACM}, 2018.
\newblock \href {https://doi.org/10.1145/3188745.3188854}
  {\path{doi:10.1145/3188745.3188854}}.

\bibitem{pace_2016_report}
Holger Dell, Thore Husfeldt, Bart M.~P. Jansen, Petteri Kaski, Christian
  Komusiewicz, and Frances~A. Rosamond.
\newblock {The First Parameterized Algorithms and Computational Experiments
  Challenge}.
\newblock In Jiong Guo and Danny Hermelin, editors, {\em 11th International
  Symposium on Parameterized and Exact Computation (IPEC 2016)}, volume~63 of
  {\em Leibniz International Proceedings in Informatics (LIPIcs)}, pages
  30:1--30:9, Dagstuhl, Germany, 2017. Schloss Dagstuhl--Leibniz-Zentrum fuer
  Informatik.
\newblock URL: \url{http://drops.dagstuhl.de/opus/volltexte/2017/6931}, \href
  {https://doi.org/10.4230/LIPIcs.IPEC.2016.30}
  {\path{doi:10.4230/LIPIcs.IPEC.2016.30}}.

\bibitem{pace_2017_report}
Holger Dell, Christian Komusiewicz, Nimrod Talmon, and Mathias Weller.
\newblock {The PACE 2017 Parameterized Algorithms and Computational Experiments
  Challenge: The Second Iteration}.
\newblock In Daniel Lokshtanov and Naomi Nishimura, editors, {\em 12th
  International Symposium on Parameterized and Exact Computation (IPEC 2017)},
  volume~89 of {\em Leibniz International Proceedings in Informatics (LIPIcs)},
  pages 30:1--30:12, Dagstuhl, Germany, 2018. Schloss Dagstuhl--Leibniz-Zentrum
  fuer Informatik.
\newblock URL: \url{http://drops.dagstuhl.de/opus/volltexte/2018/8558}, \href
  {https://doi.org/10.4230/LIPIcs.IPEC.2017.30}
  {\path{doi:10.4230/LIPIcs.IPEC.2017.30}}.

\bibitem{eppstein_cliques}
David Eppstein, Maarten L{\"{o}}ffler, and Darren Strash.
\newblock Listing all maximal cliques in large sparse real-world graphs.
\newblock {\em {ACM} J. Exp. Algorithmics}, 18, 2013.
\newblock \href {https://doi.org/10.1145/2543629} {\path{doi:10.1145/2543629}}.

\bibitem{gurobi}
{Gurobi Optimization, LLC}.
\newblock {Gurobi Optimizer Reference Manual}, 2023.
\newblock URL: \url{https://www.gurobi.com}.

\bibitem{johnson_enumerating_cliques}
David~S. Johnson, Christos~H. Papadimitriou, and Mihalis Yannakakis.
\newblock On generating all maximal independent sets.
\newblock {\em Inf. Process. Lett.}, 27(3):119--123, 1988.
\newblock \href {https://doi.org/10.1016/0020-0190(88)90065-8}
  {\path{doi:10.1016/0020-0190(88)90065-8}}.

\bibitem{bak_hyperbolic_intersection}
S{\'{a}}ndor Kisfaludi{-}Bak.
\newblock Hyperbolic intersection graphs and (quasi)-polynomial time.
\newblock In Shuchi Chawla, editor, {\em Proceedings of the 2020 {ACM-SIAM}
  Symposium on Discrete Algorithms, {SODA} 2020, Salt Lake City, UT, USA,
  January 5-8, 2020}, pages 1621--1638. {SIAM}, 2020.
\newblock \href {https://doi.org/10.1137/1.9781611975994.100}
  {\path{doi:10.1137/1.9781611975994.100}}.

\bibitem{hyp_geom_compl_ntwrks}
Dmitri Krioukov, Fragkiskos Papadopoulos, Maksim Kitsak, Amin Vahdat, and
  Mari\'an Bogu\~n\'a.
\newblock Hyperbolic geometry of complex networks.
\newblock {\em Phys. Rev. E}, 82:036106, Sep 2010.
\newblock URL: \url{https://link.aps.org/doi/10.1103/PhysRevE.82.036106}, \href
  {https://doi.org/10.1103/PhysRevE.82.036106}
  {\path{doi:10.1103/PhysRevE.82.036106}}.

\bibitem{Lovasz_2005}
László Lovász.
\newblock Graph minor theory.
\newblock {\em Bulletin of the American Mathematical Society}, 43(1):75–86,
  Oct 2005.
\newblock \href {https://doi.org/10.1090/S0273-0979-05-01088-8}
  {\path{doi:10.1090/S0273-0979-05-01088-8}}.

\bibitem{DBLP:conf/icdt/ManiuSJ19}
Silviu Maniu, Pierre Senellart, and Suraj Jog.
\newblock An experimental study of the treewidth of real-world graph data.
\newblock In Pablo Barcel{\'{o}} and Marco Calautti, editors, {\em 22nd
  International Conference on Database Theory, {ICDT} 2019, March 26-28, 2019,
  Lisbon, Portugal}, volume 127 of {\em LIPIcs}, pages 12:1--12:18. Schloss
  Dagstuhl - Leibniz-Zentrum f{\"{u}}r Informatik, 2019.
\newblock \href {https://doi.org/10.4230/LIPIcs.ICDT.2019.12}
  {\path{doi:10.4230/LIPIcs.ICDT.2019.12}}.

\bibitem{robertson_seymour_graphminors_iv}
Neil Robertson and Paul~D. Seymour.
\newblock Graph minors. {IV.} tree-width and well-quasi-ordering.
\newblock {\em J. Comb. Theory, Ser. {B}}, 48(2):227--254, 1990.
\newblock \href {https://doi.org/10.1016/0095-8956(90)90120-O}
  {\path{doi:10.1016/0095-8956(90)90120-O}}.

\bibitem{DBLP:conf/aaai/RossiA15}
Ryan~A. Rossi and Nesreen~K. Ahmed.
\newblock The network data repository with interactive graph analytics and
  visualization.
\newblock In Blai Bonet and Sven Koenig, editors, {\em Proceedings of the
  Twenty-Ninth {AAAI} Conference on Artificial Intelligence, January 25-30,
  2015, Austin, Texas, {USA}}, pages 4292--4293. {AAAI} Press, 2015.
\newblock URL:
  \url{http://www.aaai.org/ocs/index.php/AAAI/AAAI15/paper/view/9553}.

\end{thebibliography}

\appendix

\clearpage

\section{Limits of the set cover heuristics}
\label{sec:additional-stuff}

\begin{figure}[]
  \centering
  \begin{subfigure}[b]{0.48\textwidth}
  \centering
    \includegraphics[scale=0.9]{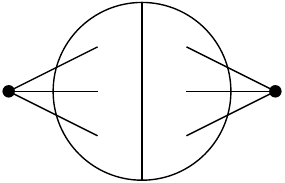}
    \caption{Unweighted set cover.}
    \label{fig:set_cover_non_opt_unweighted}
  \end{subfigure}
  ~ 
  \begin{subfigure}[b]{0.48\textwidth}
  \centering
    \includegraphics[scale=0.8]{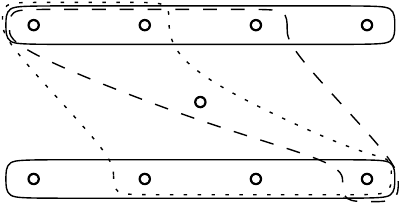}
    \caption{Weighted set cover.}
    \label{fig:set_cover_non_opt_weighted}
  \end{subfigure}
  \caption{Counter-examples for the optimality of the set
    cover heuristics.}%
  \label{fig:set_cover_non_opt}
\end{figure}

We want to briefly discuss why the set cover solutions are
not always optimal clique partitions. First, we give an
instance on which the unweighted set cover approach fails.

\begin{observation}
  There are graphs on which the minimum size clique cover
  cannot give an optimal clique partition.
\end{observation}
\begin{proof}
  Consider a clique on $k$ vertices for even $k$ where half
  of the vertices are connected to one additional vertex
  and the other half to another additional vertex, as
  illustrated in \cref{fig:set_cover_non_opt_unweighted}.  Then, for
  $k \ge 6$ the partition into three cliques of sizes 1, 1,
  and $k$ has lower weight than the partition into two
  cliques of size $\frac{k}{2}+1$, which corresponds to the
  optimal solution of the set cover instance.
\end{proof}

For the minimum weight set cover, we can use the fact that
the weights of the set cover instance correspond to the
size of the whole clique and do not reflect the potential
overlap between multiple selected cliques.

\begin{observation}
  There are graphs on which the minimum size clique cover
  cannot give an optimal clique partition.
\end{observation}
\begin{proof}
  For the weighted approach, consider the instance depicted
  in \cref{fig:set_cover_non_opt_weighted}.  The small
  circles represent the vertices of a graph and the regions
  mark maximal cliques.  The optimal clique partitioning
  uses cliques of sizes $6$, $2$, and $1$ (the dotted
  clique plus the remainders of the two solid cliques).  In
  the set cover instance these cliques have (partly
  overlapping) sizes $6$, $4$, and $4$, which is more
  expensive than the set cover solution with sizes $5$,
  $4$, and $4$ (using the dashed clique instead of the
  dotted one), which results in a solution with sizes
  $5, 3, 1$.
\end{proof}

The above problem could be avoided by extending the set
cover instance to also include all non-maximal subsets of
each clique that can be obtained by removing vertices
shared with some subset of overlapping cliques.  This
would, however, lead to an exponential blowup of the
set cover instances, which is not feasible even with
state of the art solvers.

\end{document}